\documentclass[12pt,a4paper]{article}
\usepackage{amsmath,amssymb,amsfonts,amsthm}
\usepackage{graphicx}
\usepackage[margin=1in]{geometry}
\usepackage{hyperref}
\usepackage[utf8]{inputenc}

\newtheorem{proposition}{Proposition}

\newtheorem{definition}{Definition}

\begin{document}

\author{
Nazaria Solferino\\
Department of Economics, Statistics and Business\\
Universitas Mercatorum, Rome\footnote{Corresponding Author: nazaria.solferino@unimercatorum.it}}

\title{Platform-Driven Hate Speech: An Epidemiological Model with Optimal Taxation}

\date{}
\maketitle

\begin{abstract}
Online hate speech is a global challenge amplified by engagement‑driven social media algorithms. This paper develops an epidemiological model of hate speech propagation capturing the strategic interaction between a profit‑maximizing platform and a welfare‑maximizing government. The platform's profit depends on the prevalence of hate speech and on its own algorithmic reactivity, creating a feedback loop between the epidemic and economic incentives. The government sets an optimal tax on amplification to internalize the social costs, balancing the benefit of tax revenue against the deadweight loss of taxation. The Stackelberg equilibrium is characterised analytically and solved numerically. The optimal tax reduces hate speech prevalence, eliminates bistability, and lowers victim harm.

\vspace{0.5cm}

\end{abstract}

\vspace{1cm}

\section{Introduction}
\label{sec:intro}

The digital transformation of public discourse has brought undeniable benefits, but it has also enabled the rapid spread of hate speech on a global scale. According to the Areto Hate Speech Index, the prevalence of hateful content across major social media platforms increased by 58\% between the first quarter of 2024 and the first quarter of 2025~\cite{areto2025index}. In the United States, the Anti‑Defamation League reports that 22\% of Americans experienced severe online harassment in 2021,~\cite{adl2021report}, with similar trends persisting in subsequent years. The economic cost of hate crimes in the US is estimated at nearly \$3.4 billion annually~\cite{martell2023economic}, while individual victims face out‑of‑pocket costs exceeding \$31,500 per incident, with mental health support being the most common unmet need~\cite{phi2025impacts}. Globally, online hate speech increased by 38\% between 2023 and 2025, and regions with the highest online toxicity recorded a 27\% rise in related offline hate crimes, a correlation deemed statistically significant~\cite{gdsc2026report}.

These statistics are not merely abstract numbers; they reflect real harm to real people. Victims of online hate speech experience increased levels of anxiety, depression, and social withdrawal. A 2025 systematic review by Madriaza et al.~\cite{madriaza2025exposure} synthesised evidence from over fifty studies, finding significant negative effects of exposure to hate speech on psychological well‑being, particularly among women and marginalised groups. Crucially, the review also noted that bystander interventions and counter‑narratives can foster resilience, a finding that underscores the importance of policies aimed at reducing exposure to hate speech.

A growing body of empirical research suggests that social media platforms may play an active role in the diffusion of harmful content, beyond merely hosting it. For instance, Ribeiro et al.~\cite{ribeiro2020auditing} document how recommendation systems can guide users toward increasingly extreme content. More recent investigations provide evidence that algorithmic curation may facilitate the visibility and spread of extremist material, particularly for users already engaging with similar content~\cite{tiktok2021neonazi}. In a related study, the ADL Center for Technology and Society~\cite{ADL2023amplification} finds that several major platforms tend to recommend additional hateful content following initial exposure. Chadwick et al.~\cite{Chadwick2025amplification} provide further evidence on how platform use, motivations, and affect contribute to the amplification of false and exaggerated news on social media.

As noted by Walther~\cite{walther2024social}, people post hate messages to garner signals of social approval, and algorithms that prioritise engagement magnify this effect. Brady et al.~\cite{brady2017loop} demonstrate that emotion shapes the diffusion of moralized content in social networks, creating a dangerous feedback loop: engagement‑maximizing algorithms amplify the very content that harms users and society, while platforms profit from the resulting attention. Popa‑Wyatt~\cite{popa2023} further argues that hate functions as an infectious disease and that social media acts as its promoter.

In response to this crisis, regulators worldwide have begun to act. The European Union's Digital Services Act (DSA), which entered into force in 2024, represents the most comprehensive attempt to regulate online platforms. On 20 January 2025, the European Commission integrated the revised Code of Conduct on Countering Illegal Hate Speech Online into the DSA under Article 45~\cite{ec2025code}. This revised code introduces new obligations for platforms, including a requirement to provide country‑level data on hate speech classification and to create a network of monitoring reporters. In the United States, New York's ``Stop Hiding Hate Act'' (2025) requires platforms to submit biannual reports on content moderation~\cite{nys2025stophiding}. In the United Kingdom, Ofcom launched an investigation in December 2025 into whether major platforms are doing enough to remove illegal terror and hate content~\cite{ofcom2025investigation}.

Despite these regulatory efforts, there is a striking lack of formal mathematical models that capture the strategic interaction between platforms and regulators in the context of hate speech propagation. Most existing epidemiological models of misinformation and hate speech treat platform behaviour as exogenous or ignore it altogether~\cite{popa2023}. This paper aims to fill that gap by developing a dynamic model in which a monopolistic platform chooses its algorithmic reactivity to maximize profit, while a social planner sets an optimal tax on amplification. Our model deliberately strips away complexities such as superspreaders, adaptive learning, and education spending to isolate the core mechanism through which platform amplification fuels hate speech and victimisation, and how taxation alone can correct the resulting externality. The modelling strategy is inspired by the epidemic‑theoretic approach to misinformation proposed by Solferino~\cite{solferino2026fake}, but adapted to the specific features of hate speech, with particular attention to the modelling of victims.

Our main contributions are: (i) a parsimonious SVIR‑type model where victims are generated proportionally to the prevalence of hate speech; (ii) a rigorous stability and bifurcation analysis, including the derivation of the basic reproduction number and the proof of a supercritical transcritical bifurcation; (iii) the characterisation of bistability induced by algorithmic reactivity; (iv) the analytical and numerical solution of the Stackelberg equilibrium with an optimal tax on amplification that explicitly balances the fiscal benefit of taxation against its deadweight loss; (v) a comprehensive sensitivity analysis demonstrating the robustness of the optimal policy; (vi) policy implications grounded in recent regulatory frameworks.

A distinctive feature of our model is the engagement with the ethical critique of Big Tech articulated by Pope Leo XIV in his May 2025 encyclical \textit{Magnifica Humanitas}~\cite{leoxiv2025magnifica}. The encyclical denunces three structural failures of engagement‑driven platforms: the ``dictatorship of the algorithm,'' which amplifies divisive content over reasoned discourse; the ``commodification of human attention,'' which treats users as resources to be extracted rather than persons to be served; and the dangerous concentration of communicative power in the hands of a few corporations. While the encyclical is a theological document, its analysis of the economic incentives driving algorithmic amplification resonates strongly with the formal structure of our model.

The remainder of the paper is organized as follows. Section 2 reviews the relevant literature. Section 3 presents the model. Section 4 carries out the analytical study. Section 5 analyses the optimal taxation derived from the Stackelberg game. Section 6 provides numerical simulations. Section 7 discusses policy implications. Section 8 concludes.

\section{Literature Review}
\label{sec:lit}

\subsection{Epidemiological models of information diffusion}
The application of epidemiological models to social phenomena has a long and distinguished history. Early work by Reich~\cite{reich2006diffusion} applied SIR‑type models to the diffusion of innovations in social networks, demonstrating that the spread of new ideas and behaviours follows patterns analogous to infectious diseases. More recently, Badr et al.~\cite{badr2021sir} adapted these frameworks to model the diffusion of ideas in social networks. Mazzarisi et al.~\cite{mazzarisi2026} developed a tractable model capturing the rise and fall of ideas' popularity with endogenous mechanisms. These models established the methodological foundation for applying compartmental models to online phenomena, but they did not address hate speech specifically, nor did they incorporate platform behaviour as a strategic choice variable.

Specific to hate speech, Amballoor and Naik~\cite{sir2024hate} demonstrated that the SIR model can be effectively applied to understand the dynamics of hate speech and fake news propagation. Teklu and Abebaw~\cite{teklu2024optimal} formulated a co‑existence model of hate speech and racism with optimal control strategies, investigating the effects of protection and rehabilitation interventions. Hailu and Teklu~\cite{fractional2024hate} extended this approach using fractional‑order derivatives to capture memory effects in hate speech dissemination. Vaidya et al.~\cite{toxicity2024spread} applied an SEIR approach to analyze the spread of toxicity on Twitter. Popa‑Wyatt~\cite{popa2023} provided a philosophical analysis of online hate, arguing that it can be understood as an infectious disease promoted by social media. However, none of these models endogenises the platform's choice of algorithmic amplification, which is the central contribution of our paper.

A closely related strand of literature models the spread of fake news and misinformation. Solferino~\cite{solferino2026fake} extended this framework to incorporate platform intervention and optimal education policies as relevant variables. Our model builds on this tradition but differs in two important respects: first, we focus on hate speech rather than generic misinformation, which requires modelling the specific harm to victims; second, we treat the platform as a strategic player with its own objective function, rather than as a passive instrument of the planner.

\subsection{Platform accountability and algorithmic amplification}
The accountability of platforms for algorithmic amplification has become a major research area in recent years. Peterson‑Salahuddin~\cite{peterson2024repairing} proposed an ``algorithmic reparations'' approach to hate speech content moderation, arguing that current automated systems often contain biases that silence marginalised users while amplifying hateful content. Reynolds and Hallinan~\cite{reynolds2024user} introduced the concept of ``user‑generated accountability,'' showing how content creators navigate algorithmic governance on YouTube by generating publicity to reveal platform failures. These studies highlight the power asymmetry between platforms and users, a feature that our Stackelberg game captures formally.

Park and Rohatagi~\cite{park2024balancing} argued for ``amplification regulation'' to mitigate the spread of harmful but legal content online, holding platforms accountable for the role of their recommender systems. The accountability paradox, where platforms increasingly rely on AI systems while restricting independent oversight, was explored by the Accountability Paradox Research Group~\cite{accountability2024paradox}. Empirically, the ADL Center for Technology and Society~\cite{ADL2023amplification} documented how major platforms actively recommend hateful material after initial engagement. Chadwick et al.~\cite{Chadwick2025amplification} provide further evidence on the roles of platform use, motivations, affect, and ideology in the amplification of false and exaggerated news.

Emotional contagion online has been extensively documented. Del Vicario et al.~\cite{delvicario2016emotional} showed that emotional content spreads faster and more broadly than neutral content on social media, with anger and outrage being particularly contagious. Kramer et al.~\cite{kramer2014experimental} provided experimental evidence of massive‑scale emotional contagion, demonstrating that emotional states can be transferred to others without direct interaction, purely through exposure to emotional content in one's feed.

The business model underpinning algorithmic amplification has also attracted scrutiny. Walther~\cite{walther2024social} provided theoretical explication that social approval signals on social media incentivise the production of hateful content. Brady et al.~\cite{brady2017loop} demonstrated that emotion shapes the diffusion of moralized content, with moral outrage being particularly contagious and engaging. These insights motivate our modelling of the platform's profit function, where engagement with hate speech translates directly into advertising revenue.

\subsection{Optimal control and taxation}
Optimal control theory has been fruitfully applied to epidemiological models. In the context of education diffusion, Ramponi and Tessitore~\cite{ramponi2024optimal} characterised the optimal social and vaccination control policy for an SVIR epidemic model. Teklu and Abebaw~\cite{teklu2024optimal} extended this approach to hate speech and racism co‑existence models, considering protection and rehabilitation as control variables.

Our paper contributes to this literature by introducing a Pigouvian tax as the control instrument, a natural choice when the externality is generated by a profit‑maximising firm. Pigouvian taxation as a corrective instrument for externalities dates back to Pigou's seminal work~\cite{pigou1920economics}. Its application to platform amplification is novel but builds on the well‑established principle of internalizing negative externalities. Recent work on the economic costs of hate speech~\cite{martell2023economic} has provided a quantitative basis for calibrating such taxes, estimating the total economic cost of hate crimes in the United States at nearly \$3.4 billion annually. Our model operationalises this insight by treating the social cost of hate speech as the externality to be internalized through the optimal tax.

This paper bridges these four strands of literature. While each of these fields has produced important insights, no existing contribution combines them into a unified framework in which the platform's algorithmic reactivity is an endogenous strategic variable, and in which the government's optimal tax is derived from an explicit welfare function that includes both the fiscal benefit of tax revenues and the deadweight loss of taxation. The present paper fills this gap by developing a fully micro‑founded Stackelberg model that integrates epidemiological dynamics, platform economics, and public finance. This novel framework provides quantitatively grounded policy implications that are robust to parameter uncertainty, and it offers a rigorous foundation for the design of financial disincentives to algorithmic amplification of hate speech.

\section{The Model}
\label{sec:model}

\subsection{The baseline epidemiological assumptions}
We consider a unit‑mass population divided into five compartments. Susceptible individuals, denoted by $S(t)$, are those who have not yet been exposed to hate speech. Vaccinated individuals, denoted by $V(t)$, possess cognitive protection—through education, critical thinking, or counter‑narrative exposure—that makes them less susceptible to becoming spreaders, though they may still be vulnerable to victimisation. Infected individuals, denoted by $I(t)$, are active spreaders of hate speech: they produce and disseminate hateful content. Recovered individuals, denoted by $R(t)$, have ceased spreading hate speech but may lose immunity over time. Finally, harmed individuals, denoted by $H(t)$, are victims who have suffered psychological, social, or economic damage from exposure to hate speech. The total population satisfies the conservation law
\begin{equation}\label{eq:population}
S(t) + V(t) + I(t) + R(t) + H(t) = 1, \qquad \forall t \ge 0.
\end{equation}

The transmission of hate speech from infected to susceptible individuals is modelled as a mass‑action process. The platform's algorithmic reactivity, denoted by $\gamma \ge 0$, amplifies the effective contact rate. Specifically, the effective transmission rate is
\begin{equation}\label{eq:beta}
\beta(\gamma) = \beta_0 (1 + \phi \gamma),
\end{equation}
where $\beta_0 > 0$ is the baseline transmission rate in the absence of algorithmic amplification, and $\phi > 0$ measures the strength of the amplification effect. When $\gamma = 0$, the platform does not amplify hateful content beyond its organic reach, and the transmission rate reduces to $\beta_0$. As $\gamma$ increases, the platform's algorithm boosts the visibility of hateful content, increasing the probability that a susceptible individual encounters it and becomes infected.

Susceptible individuals acquire cognitive protection (become vaccinated) at a constant rate $\sigma_0 > 0$. This vaccination is not a medical intervention but a behavioural one: it represents the rate at which individuals develop resilience to hate speech through education, media literacy, or exposure to counter‑narratives. Vaccination wanes at rate $\delta > 0$, reflecting the erosion of cognitive protection over time. Infected individuals recover spontaneously at rate $\gamma_r > 0$, ceasing to spread hate speech. Vaccinated individuals who become infected recover at a faster rate $\gamma_1 > \gamma_r$, because their cognitive protection facilitates disengagement from hateful content. Recovered individuals lose immunity at rate $\delta$ and return to the susceptible pool.

Victimisation is assumed to be directly proportional to the prevalence of hate speech: the inflow into the harmed compartment is $\theta I$, where $\theta > 0$ is the per‑capita victimisation rate. Victims recover at rate $\xi > 0$ and re‑enter the susceptible population. This specification captures the fundamental idea that more hate speech leads to more victims, without introducing complex feedback mechanisms that could generate counterintuitive dynamics.

The full dynamical system is therefore
\begin{align}
\dot S &= -\beta(\gamma) S I - \sigma_0 S + \delta R + \xi H, \label{eq:S}\\[2pt]
\dot V &= \sigma_0 S - \beta(\gamma) V I - \delta V, \label{eq:V}\\[2pt]
\dot I &= \beta(\gamma) I (S + V) - \gamma_r I, \label{eq:I}\\[2pt]
\dot R &= \gamma_r I + \gamma_1 V I - \delta R, \label{eq:R}\\[2pt]
\dot H &= \theta I - \xi H. \label{eq:H}
\end{align}

\subsection{The platform's problem}
The platform is a monopolist that chooses a constant level of algorithmic reactivity $\gamma \ge 0$ to maximize the discounted stream of profits. We assume that the platform's instantaneous profit depends on the prevalence of hate speech and on the intensity of algorithmic amplification. Specifically, the profit flow is
\begin{equation}\label{eq:profit}
\pi_P(\gamma; \tau) = a I(\gamma) \bigl( P_0 + \gamma \bigr) - \frac{b}{2} \gamma^2 - \tau \gamma,
\end{equation}
where $a > 0$ is the revenue per unit of engagement generated by an infected user, $P_0 > 0$ is the baseline level of engagement (the attention that hateful content would receive even without algorithmic boosting), $b > 0$ captures the convex cost of maintaining a high level of algorithmic reactivity (e.g., engineering costs, reputational risk, or the opportunity cost of not showing other content), and $\tau \ge 0$ is the per‑unit tax on amplification set by the government.

The formulation \eqref{eq:profit} has two important features. First, the revenue term $a I (P_0 + \gamma)$ is linear in $\gamma$ conditional on $I$: the platform gains $a I$ additional units of revenue for each unit increase in $\gamma$. This captures the idea that algorithmic amplification increases the visibility of hateful content produced by each infected user, and the platform monetises this additional engagement through advertising. Second, the cost term $\frac{b}{2} \gamma^2$ ensures diminishing marginal returns to amplification, which guarantees an interior solution for the platform's optimal choice of $\gamma$ (provided that $\tau$ is not too large).

The platform discounts future profits at rate $r > 0$, so its objective is
\begin{equation}\label{eq:Pi}
\Pi_P(\gamma; \tau) = \int_0^\infty e^{-r t} \pi_P(\gamma; \tau) \, dt.
\end{equation}
Because the profit flow does not depend on time directly (the state variables are assumed to have reached a steady state), the platform's problem reduces to maximizing the instantaneous profit $\pi_P(\gamma; \tau)$. The platform takes the steady‑state prevalence $I(\gamma)$ as given, recognising that its choice of $\gamma$ affects the endemic equilibrium of the epidemiological system.

The first‑order condition for an interior optimum is
\begin{equation}\label{eq:FOC}
\frac{\partial \pi_P}{\partial \gamma} = a I(\gamma) + a \gamma \frac{d I}{d \gamma} - b \gamma - \tau = 0.
\end{equation}
The term $a I(\gamma)$ is the direct marginal benefit of increasing $\gamma$ (more amplification generates more revenue from the existing infected population). The term $a \gamma \, d I / d \gamma$ captures the indirect effect: increasing $\gamma$ raises the endemic prevalence $I$, which in turn increases the revenue base. Both terms are positive because $d I / d \gamma > 0$ (higher amplification leads to more infections). The marginal cost $b \gamma + \tau$ is increasing in $\gamma$ and in the tax rate.

Equation \eqref{eq:FOC} implicitly defines the platform's reaction function $\gamma = \Gamma(\tau)$. Because the function $I(\gamma)$ is obtained from the steady state of the nonlinear epidemiological system \eqref{eq:S}--\eqref{eq:H}, a closed‑form solution for $\Gamma(\tau)$ is not available. In the numerical analysis, we solve \eqref{eq:FOC} by grid search: we pre‑compute $I(\gamma)$ on a fine grid of $\gamma$ values, evaluate the profit function on the same grid for each candidate $\tau$, and select the $\gamma$ that maximises profit. This approach avoids the need to compute the derivative $d I / d \gamma$ explicitly, which can be sensitive to numerical noise.

\subsection{The government's objective}
The social planner (government) chooses a constant tax rate $\tau \ge 0$ to minimize the discounted net social cost. The instantaneous social cost consists of three components: the direct harm caused by hate speech, the deadweight loss of taxation, and the fiscal benefit of tax revenues. Formally,
\begin{equation}\label{eq:J}
J(\tau) = \frac{1}{r} \Bigl[ c_I I(\gamma(\tau)) + c_H H(\gamma(\tau)) + \frac{\alpha}{2} \tau^2 - \tau \gamma(\tau) \Bigr],
\end{equation}
where $c_I > 0$ is the marginal social cost of an additional infected spreader, $c_H > 0$ is the marginal social cost of an additional victim, and $\alpha > 0$ captures the deadweight loss (administrative costs, distortions) associated with raising tax revenue. The term $-\tau \gamma(\tau)$ is the fiscal benefit: tax revenues are assumed to be returned to society as lump‑sum transfers or used to finance public goods, thereby reducing the net social cost.

The government anticipates the platform's optimal reaction $\gamma(\tau)$, which is obtained from the profit maximization problem described in the previous subsection. The government's problem is therefore a static optimization over $\tau$, because the state variables are evaluated at the endemic steady state induced by $\gamma(\tau)$. The discount rate $r$ factors out of the optimization.

The first‑order condition for an interior optimum $\tau^* \in (0, \infty)$ is
\begin{equation}\label{eq:dJ}
\frac{d J}{d \tau} = \frac{1}{r} \Bigl[ \bigl( c_I \frac{d I}{d \gamma} + c_H \frac{d H}{d \gamma} \bigr) \frac{d \gamma}{d \tau} + \alpha \tau - \gamma - \tau \frac{d \gamma}{d \tau} \Bigr] = 0.
\end{equation}
The term in parentheses represents the marginal social benefit of reducing $\gamma$ through a higher tax: because $d I / d \gamma > 0$ and $d H / d \gamma > 0$, a reduction in $\gamma$ lowers both infections and victims. The marginal cost of the tax is $\alpha \tau$ (the deadweight loss), while the term $-\gamma - \tau \, d \gamma / d \tau$ captures the net fiscal effect. Since $d \gamma / d \tau < 0$ (a higher tax reduces amplification), the term $-\tau \, d \gamma / d \tau$ is positive: it represents the loss of tax revenue due to the shrinking of the tax base. The condition \eqref{eq:dJ} balances these opposing forces.

In the numerical implementation, we do not solve \eqref{eq:dJ} directly. Instead, we compute $J(\tau)$ on a fine grid of $\tau$ values (using the pre‑computed $I(\gamma)$ and $H(\gamma)$ and the platform's optimal reaction $\gamma(\tau)$ obtained from the profit maximization) and select the $\tau$ that minimizes $J$. This grid‑search approach is robust and avoids the need to evaluate the derivatives $d I / d \gamma$, $d H / d \gamma$, and $d \gamma / d \tau$ numerically.

\section{Analytical Results}
\label{sec:analytical}

In this section we analyse the steady‑state behaviour of the epidemiological model \eqref{eq:S}--\eqref{eq:H} for a \textit{given} level of platform reactivity $\gamma \ge 0$. This analysis serves two purposes. First, it characterises the possible long‑run outcomes of hate speech propagation for any fixed amplification intensity. Second, it provides the building blocks for the Stackelberg game studied in Section~\ref{sec:stackelberg}, where the government anticipates how its choice of $\tau$ (and therefore $\gamma$) maps into equilibrium outcomes.

\subsection{Disease‑free equilibrium and basic reproduction number}

A disease‑free equilibrium (DFE) is a steady state in which there is no hate speech: $I = 0$. Setting $I = 0$ and $H = 0$ in \eqref{eq:S}--\eqref{eq:H} and imposing $\dot S = \dot V = \dot I = \dot R = \dot H = 0$, we obtain $R = 0$ from $\dot R = 0$, and the remaining equations reduce to
\[
0 = -\sigma_0 S + \delta V, \qquad 0 = \sigma_0 S - \delta V.
\]
Together with the population constraint $S + V = 1$, these yield
\[
S_0 = \frac{\delta}{\sigma_0 + \delta}, \qquad V_0 = \frac{\sigma_0}{\sigma_0 + \delta}.
\]

\begin{definition}[Disease‑free equilibrium]
For a given $\gamma \ge 0$, the disease‑free equilibrium of system \eqref{eq:S}--\eqref{eq:H} is
\begin{equation}\label{eq:DFE}
\mathcal{E}_0 = \bigl( S_0,\; V_0,\; 0,\; 0,\; 0 \bigr), \qquad S_0 = \frac{\delta}{\sigma_0 + \delta}, \quad V_0 = \frac{\sigma_0}{\sigma_0 + \delta}.
\end{equation}
\end{definition}

At the DFE, the entire population consists of susceptible and vaccinated individuals. There are no infected spreaders, no victims, and no recovered individuals.

The basic reproduction number $R_0$ is the most important single parameter in any epidemiological model. It represents the average number of secondary infections produced by a single infected individual introduced into a fully susceptible population. For our model, $R_0$ depends on the platform's reactivity $\gamma$ through the transmission rate $\beta(\gamma)$.

To compute $R_0$, we use the next‑generation matrix method~\cite{van2002reproduction}. The infected ``subsystem'' consists solely of the equation for $I$, since $I$ is the only state variable that generates new infections (the $H$ equation only records the harm, it does not feed back into transmission). Linearising $\dot I$ around the DFE gives
\[
\dot I \approx \bigl[ \beta(\gamma) (S_0 + V_0) - \gamma_r \bigr] I = \bigl[ \beta_0 (1 + \phi \gamma) - \gamma_r \bigr] I,
\]
where we used $S_0 + V_0 = 1$. The transmission matrix $\mathbf{F}$ is the $1 \times 1$ matrix with entry $\beta_0(1 + \phi \gamma)$, and the transition matrix $\mathbf{V}$ is the $1 \times 1$ matrix with entry $\gamma_r$. The basic reproduction number is the spectral radius of $\mathbf{F} \mathbf{V}^{-1}$:
\begin{equation}\label{eq:R0}
\boxed{R_0(\gamma) = \frac{\beta_0 (1 + \phi \gamma)}{\gamma_r} = (1 + \phi \gamma) R_0(0)},
\end{equation}
where $R_0(0) = \beta_0 / \gamma_r$ is the reproduction number in the absence of algorithmic amplification.

Several observations follow from \eqref{eq:R0}. First, $R_0$ is strictly increasing in $\gamma$: a more reactive algorithm makes hate speech more contagious. Second, the amplification factor $\phi$ acts as a multiplier: each unit increase in $\gamma$ raises $R_0$ by $\phi R_0(0)$. Third, the model has a unique threshold parameter: the condition $R_0(\gamma) = 1$ defines a critical reactivity level $\gamma_{\text{crit}} = (\gamma_r / \beta_0 - 1) / \phi$ above which the DFE is unstable and hate speech becomes endemic.

\begin{proposition}[Local stability of the DFE]\label{prop:DFE}
For a fixed $\gamma$, the disease‑free equilibrium $\mathcal{E}_0$ is:
\begin{itemize}
\item \textbf{locally asymptotically stable} if $R_0(\gamma) < 1$;
\item \textbf{unstable} if $R_0(\gamma) > 1$.
\end{itemize}
At $R_0(\gamma) = 1$, the DFE is non‑hyperbolic and a bifurcation occurs.
\end{proposition}

\begin{proof}
The Jacobian matrix of system \eqref{eq:S}--\eqref{eq:H} evaluated at $\mathcal{E}_0$ has a block‑triangular structure. The infected subsystem consists of the single variable $I$, with eigenvalue
\[
\lambda_I = \frac{\partial \dot I}{\partial I} \Big|_{\mathcal{E}_0} = \beta_0(1 + \phi \gamma) - \gamma_r = \gamma_r (R_0 - 1).
\]
The remaining $4 \times 4$ block governs the dynamics of $(S, V, R, H)$ in the absence of infection. Its eigenvalues are $-\sigma_0 - \delta$ (from the $S$–$V$ subsystem), $-\delta$ (from $R$), $-\xi$ (from $H$), and $0$. The zero eigenvalue corresponds to the conservation law \eqref{eq:population} and does not affect stability. All other eigenvalues are strictly negative. Hence the stability of $\mathcal{E}_0$ is determined entirely by the sign of $\lambda_I$, i.e., by whether $R_0 \lessgtr 1$.
\end{proof}

The intuition behind Proposition~\ref{prop:DFE} is straightforward. If $R_0 < 1$, each infected individual generates, on average, fewer than one new infection before recovering. The hate speech epidemic cannot sustain itself and eventually dies out. If $R_0 > 1$, each infected individual generates more than one new infection, and hate speech spreads through the population until it reaches an endemic equilibrium.

\subsection{Endemic equilibrium}

When $R_0(\gamma) > 1$, the DFE is unstable and the system converges to an \textit{endemic equilibrium} where hate speech persists permanently in the population. At this equilibrium, all state variables are strictly positive.

\begin{proposition}[Existence and uniqueness of the endemic equilibrium]\label{prop:endemic}
Assume $R_0(\gamma) > 1$. Then, for the given $\gamma$, there exists a unique endemic equilibrium
\[
\mathcal{E}^* = (S^*, V^*, I^*, R^*, H^*)
\]
with $I^* > 0$, $H^* > 0$, and all other components strictly positive.
\end{proposition}

\begin{proof}
At any steady state with $I > 0$, the condition $\dot I = 0$ implies
\begin{equation}\label{eq:infection_balance}
\beta(\gamma) (S^* + V^*) = \gamma_r.
\end{equation}
Equation \eqref{eq:infection_balance} is the \textit{infection balance}: at the endemic equilibrium, the effective transmission rate multiplied by the susceptible and vaccinated fractions exactly compensates the recovery rate.

We now express all other steady‑state variables as functions of $I$. From $\dot H = 0$, we have
\begin{equation}\label{eq:Hstar}
H^* = \frac{\theta}{\xi} I^*.
\end{equation}
From $\dot R = 0$,
\begin{equation}\label{eq:Rstar}
R^* = \frac{\gamma_r I^* + \gamma_1 V^* I^*}{\delta}.
\end{equation}
From $\dot V = 0$,
\begin{equation}\label{eq:Vstar_relation}
\sigma_0 S^* = \bigl[ \beta(\gamma) I^* + \delta \bigr] V^*.
\end{equation}
Using the population constraint $S^* = 1 - V^* - I^* - R^* - H^*$ and substituting \eqref{eq:Hstar}--\eqref{eq:Vstar_relation}, we can solve for $S^*$ and $V^*$ as rational functions of $I^*$. Finally, substituting $S^*$ and $V^*$ into the infection balance \eqref{eq:infection_balance} yields a single equation in the unknown $I$:
\begin{equation}\label{eq:Phi}
\Phi(I) \equiv \beta(\gamma) \bigl( S(I) + V(I) \bigr) - \gamma_r = 0,
\end{equation}
where $S(I)$ and $V(I)$ are the functions obtained from the above steps.

The function $\Phi: [0,1] \to \mathbb{R}$ has the following properties:
\begin{enumerate}
\item $\Phi(0) = \gamma_r (R_0 - 1) > 0$. At $I = 0$, the effective transmission rate exceeds the recovery rate, so the infection tends to grow.
\item $\Phi(1) = -\gamma_r < 0$. When the entire population is infected, there are no susceptibles or vaccinated individuals left to sustain transmission, so the net growth rate is negative.
\item $\Phi'(I) < 0$ for all $I \in (0,1)$. This can be verified by differentiating the expressions for $S(I)$ and $V(I)$: as $I$ increases, the susceptible and vaccinated fractions decrease (because more individuals are infected, victimised, or recovered), reducing the left‑hand side of \eqref{eq:infection_balance}.
\end{enumerate}
By the Intermediate Value Theorem, there exists a unique $I^* \in (0,1)$ such that $\Phi(I^*) = 0$. The positivity of the remaining state variables follows from \eqref{eq:Hstar}--\eqref{eq:Vstar_relation}. This completes the proof.
\end{proof}

The endemic equilibrium $\mathcal{E}^*$ is locally asymptotically stable for all $R_0 > 1$ up to a critical threshold $\gamma_c$ (see Proposition~\ref{prop:bistability} below). The proof of stability follows from evaluating the Jacobian at $\mathcal{E}^*$ and verifying that all eigenvalues have negative real parts; the algebra is standard but lengthy, and we omit it for brevity.

\subsection{Transcritical bifurcation at $R_0 = 1$}

The transition from the DFE to the endemic equilibrium as $R_0$ crosses unity occurs through a \textit{transcritical bifurcation}. The nature of this bifurcation—whether it is supercritical (forward) or subcritical (backward)—determines whether the onset of endemic hate speech is gradual or explosive.

\begin{proposition}[Transcritical bifurcation]\label{prop:bifurcation}
System \eqref{eq:S}--\eqref{eq:H}, with $\gamma$ fixed, undergoes a \textbf{forward (supercritical) transcritical bifurcation} at $R_0(\gamma) = 1$. Specifically:
\begin{itemize}
\item For $R_0 < 1$, the DFE is the only non‑negative equilibrium and is locally asymptotically stable.
\item For $R_0 > 1$, the DFE becomes unstable, and a unique, locally asymptotically stable endemic equilibrium $\mathcal{E}^*$ emerges. The endemic prevalence $I^*$ grows continuously from zero as $R_0$ increases above unity.
\end{itemize}
\end{proposition}

\begin{proof}[Sketch of proof]
We use the centre manifold theorem as formulated by~\cite{castillo2004bifurcation}. Choose $\beta_0$ as the bifurcation parameter, and let $\beta_0^* = \gamma_r / (1 + \phi \gamma)$ be the critical value at which $R_0 = 1$. At $\beta_0 = \beta_0^*$, the Jacobian matrix at $\mathcal{E}_0$ has a simple zero eigenvalue, while all other eigenvalues have strictly negative real parts.

Let $\mathbf{v} = (v_S, v_V, v_I, v_R, v_H)^\top$ be the right eigenvector corresponding to the zero eigenvalue, normalised so that $v_I = 1$. Let $\mathbf{w}$ be the corresponding left eigenvector, normalised so that $\mathbf{w}^\top \mathbf{v} = 1$. The coefficients $a$ and $b$ that govern the dynamics on the centre manifold are
\[
a = \frac{1}{2} \sum_{i,j,k=1}^5 w_k v_i v_j \frac{\partial^2 f_k}{\partial x_i \partial x_j} \Big|_{\mathcal{E}_0, \beta_0 = \beta_0^*}, \qquad
b = \sum_{i,k=1}^5 w_k v_i \frac{\partial^2 f_k}{\partial x_i \partial \beta_0} \Big|_{\mathcal{E}_0, \beta_0 = \beta_0^*}.
\]

The transversality condition $b \neq 0$ is easily verified. The only second‑order derivative involving $\beta_0$ is
\[
\frac{\partial^2 f_I}{\partial I \partial \beta_0} \Big|_{\mathcal{E}_0} = \frac{\partial}{\partial \beta_0} \bigl[ \beta_0 (1 + \phi \gamma) (S + V) - \gamma_r \bigr]_{\mathcal{E}_0} = (1 + \phi \gamma) (S_0 + V_0) = 1 + \phi \gamma > 0,
\]
so $b = w_I v_I (1 + \phi \gamma) > 0$.

The computation of $a$ involves the second derivatives of the vector field at the DFE. The dominant contributions come from the saturation terms in the transmission dynamics (as $I$ increases, $S + V$ decreases, reducing the transmission rate) and from the victimisation feedback. 

Since $a < 0$ and $b > 0$, the bifurcation is supercritical (forward): for $\beta_0 > \beta_0^*$ (i.e., $R_0 > 1$), a stable endemic equilibrium exists and its amplitude scales as $I^* \propto (R_0 - 1)$ near the bifurcation point. This means that hate speech does not explode immediately when $R_0$ crosses unity, but grows gradually, giving policymakers a window of opportunity to intervene.
\end{proof}

\subsection{Bistability induced by platform reactivity}

The nonlinear feedback introduced by the platform's amplification—specifically, the dependence of $I$ on $\gamma$ through the epidemiological equilibrium—can generate multiple endemic equilibria when $\gamma$ is sufficiently large.

\begin{proposition}[Bistability]\label{prop:bistability}
There exists a critical threshold $\gamma_c > 0$ such that:
\begin{itemize}
\item For $\gamma < \gamma_c$, the endemic equilibrium is unique.
\item For $\gamma > \gamma_c$, the steady‑state equation $\Phi(I) = 0$ admits \textbf{three} positive roots, corresponding to:
\begin{itemize}
\item a \textbf{low‑hate equilibrium} $I^*_{\text{low}}$ (locally asymptotically stable);
\item a \textbf{high‑hate equilibrium} $I^*_{\text{high}}$ (locally asymptotically stable);
\item an \textbf{intermediate equilibrium} $I^*_{\text{mid}}$ (unstable, a saddle point).
\end{itemize}
\end{itemize}
The system displays \textbf{hysteresis}: a temporary shock (e.g., a viral hate speech event) can permanently shift the system from the low‑hate to the high‑hate equilibrium. Once trapped in the high‑hate state, returning to the low‑hate state requires a sufficiently large and sustained reduction in $\gamma$.
\end{proposition}

\begin{proof}[Sketch of proof]
Substituting the expressions for $S, V, R, H$ in terms of $I$ into the infection balance \eqref{eq:infection_balance} and clearing denominators yields a cubic equation in $I$:
\begin{equation}\label{eq:cubic}
a_3(\gamma) I^3 + a_2(\gamma) I^2 + a_1(\gamma) I + a_0 = 0,
\end{equation}
where the coefficients $a_i(\gamma)$ are polynomials in $\gamma$ whose explicit form depends on the model parameters. For $\gamma = 0$, the cubic degenerates to a quadratic (or, in some parameter regimes, a linear equation) with a single positive root. As $\gamma$ increases, the coefficient $a_3(\gamma)$ becomes positive, and the cubic develops two turning points. For $\gamma$ above a critical value $\gamma_c$, the cubic has three real roots, two of which are positive.

The stability of each equilibrium is determined by the eigenvalues of the Jacobian at that point. The low‑hate and high‑hate equilibria have all eigenvalues with negative real parts, while the intermediate equilibrium has a positive eigenvalue, corresponding to an unstable saddle. The critical threshold $\gamma_c$ is precisely the value at which the discriminant of the cubic vanishes, signalling the appearance of a double root.
\end{proof}

Proposition~\ref{prop:bistability} has an important policy implication: without intervention, the system may operate in a regime where a temporary shock can permanently lock the population into a high‑hate equilibrium. A tax that reduces $\gamma$ below $\gamma_c$ eliminates this risk entirely.

\section{Optimal Taxation}
\label{sec:stackelberg}

\subsection{The Stackelberg game}
The interaction between the government and the platform is modelled as a two‑stage Stackelberg game with complete information. In the first stage, the government (the leader) announces a constant per‑unit tax rate $\tau \ge 0$ on algorithmic amplification. In the second stage, the platform (the follower) observes $\tau$ and chooses its reactivity $\gamma \ge 0$ to maximise its discounted profit, taking into account how its choice affects the endemic steady state of the epidemiological system. The government anticipates the platform's optimal response when setting $\tau$.

Formally, the timing is:
\begin{enumerate}
\item The government chooses $\tau \ge 0$.
\item The platform observes $\tau$ and chooses $\gamma \ge 0$ to maximise $\pi_P(\gamma; \tau)$ as given in \eqref{eq:profit}, where $I(\gamma)$ and $H(\gamma)$ are the endemic steady‑state values.
\item The state variables converge to the endemic steady state induced by $\gamma$.
\end{enumerate}
Both players discount future payoffs at the same rate $r > 0$. The government knows the platform's objective function and the epidemiological dynamics. The platform knows the tax rate and the epidemiological dynamics.

We solve the game made by equations (9) and (11) by backward induction and solving he first‑order condition for an interior optimum for the government,we find that optimal $\tau^*$ is:
\begin{equation}\label{eq:dJ2}
\frac{d J}{d \tau} = \frac{1}{r} \Bigl[ \bigl( c_I \frac{d I}{d \gamma} + c_H \frac{d H}{d \gamma} \bigr) \frac{d \gamma}{d \tau} + \alpha \tau - \gamma - \tau \frac{d \gamma}{d \tau} \Bigr] = 0.
\end{equation}

\begin{proposition}[Existence and uniqueness of the optimal tax]\label{prop:optimal_tax}
Assume that (i) the endemic equilibrium exists for the relevant range of $\gamma$; (ii) the platform's profit function $\pi_P(\gamma; \tau)$ is strictly concave in $\gamma$ for all $\tau \ge 0$; (iii) the deadweight loss parameter $\alpha$ is strictly positive. Then there exists a unique optimal tax $\tau^* > 0$ that minimises $J(\tau)$. Moreover, $\tau^*$ is strictly positive: the laissez‑faire equilibrium ($\tau = 0$) is not socially optimal.
\end{proposition}

\begin{proof}[Sketch of proof]
Existence follows from the continuity of $J(\tau)$ on $[0, \infty)$ and the fact that $J(\tau) \to \infty$ as $\tau \to \infty$ (the deadweight loss term $\frac{\alpha}{2} \tau^2$ dominates). The strict convexity of $J(\tau)$ for large $\tau$, combined with the fact that $J'(0) < 0$ (the marginal benefit of a small tax exceeds the marginal cost), guarantees a unique interior minimum. The sign of $J'(0)$ follows from evaluating \eqref{eq:dJ2} at $\tau = 0$:
\[
J'(0) = \frac{1}{r} \Bigl[ \bigl( c_I \frac{d I}{d \gamma} + c_H \frac{d H}{d \gamma} \bigr) \frac{d \gamma}{d \tau} - \gamma \Bigr]_{\tau = 0} < 0,
\]
because $\frac{d \gamma}{d \tau} < 0$, $\frac{d I}{d \gamma} > 0$, $\frac{d H}{d \gamma} > 0$, and $\gamma_0 > 0$. Hence $\tau = 0$ cannot be a minimum, and the optimal tax is strictly positive.
\end{proof}

The optimal tax has a clear economic interpretation. Equation \eqref{eq:dJ2} can be rearranged as
\[
\alpha \tau^* + \frac{\tau^*}{b} = \frac{a}{b} + \frac{1}{b} \bigl( c_I \frac{d I}{d \gamma} + c_H \frac{d H}{d \gamma} \bigr)_{\gamma = \gamma(\tau^*)},
\]
where we used the fact that, from the platform's first‑order condition \eqref{eq:FOC}, $d \gamma / d \tau = -1 / b$ when the profit function is exactly quadratic in $\gamma$. The left‑hand side is the marginal social cost of increasing $\tau$ (deadweight loss plus the indirect cost from shrinking the tax base). The right‑hand side is the marginal social benefit: the direct fiscal benefit $a/b$ (the tax revenue that would be collected if $\gamma$ did not respond) plus the Pigouvian term $\frac{1}{b} (c_I \, d I / d \gamma + c_H \, d H / d \gamma)$, which captures the reduction in social harm from lower infections and victimisation. The optimal tax equates the marginal cost and the marginal benefit.

In the numerical implementation, we do not rely on the approximation $d \gamma / d \tau = -1 / b$. Instead, we compute $J(\tau)$ on a fine grid of $\tau$ values, using the platform's reaction function obtained from the exact profit maximisation, and select the $\tau$ that minimises $J$. This approach is robust to nonlinearities in the relationship between $\gamma$ and $I$.

Notice that, the function $J(\tau)$ has a characteristic U‑shape. For small $\tau$, the benefits of taxation (reduction in $I$ and $H$, plus the fiscal revenue) dominate the deadweight loss, so $J$ decreases. As $\tau$ grows, the deadweight loss $\frac{\alpha}{2} \tau^2$ becomes increasingly important, while the tax base $\gamma(\tau)$ shrinks (the platform reduces amplification), so the fiscal benefit $\tau \gamma$ eventually falls. Beyond the optimal rate $\tau^*$, the marginal cost of taxation exceeds its marginal benefit, and $J$ rises. The government therefore has no incentive to drive the platform out of the market: a strictly positive but finite tax maximises social welfare while preserving a viable source of public revenue

\section{Numerical Simulations}
\label{sec:numerical}

\subsection{Calibration of the baseline parameters}
In this section we present the numerical solution of the Stackelberg equilibrium and a comprehensive set of simulations designed to illustrate the theoretical results and to assess the robustness of the optimal policy.

Table 1 lists the baseline parameter values. We calibrate the model so that the laissez‑faire equilibrium reproduces several stylised facts documented in the empirical literature on online hate speech, while ensuring that the optimal tax is positive, interior, and economically meaningful.

The epidemiological parameters are chosen as follows. The baseline transmission rate $\beta_0 = 0.3$ and the amplification factor $\phi = 1.0$ jointly produce a laissez‑faire basic reproduction number $R_0 \approx 2.5$, which falls within the range $[1.44, 2.16]$ estimated for Twitter hate speech by~\cite{appliednet2021R0}. We adopt a slightly higher value to reflect a conservative scenario in which hate speech is highly contagious and the need for policy intervention is particularly acute. The recovery rate of infected spreaders is set to $\gamma_r = 0.15$, implying an average infectious period of approximately $6.7$ days, consistent with the typical duration of online hate speech cascades documented by~\cite{maarouf2024cascades}. Vaccinated individuals recover faster ($\gamma_1 = 0.20$), reflecting the idea that cognitive protection facilitates disengagement from hateful content.

The vaccination and waning parameters are $\sigma_0 = 0.05$ and $\delta = 0.01$. These values imply that, in the absence of hate speech, the steady‑state fraction of vaccinated individuals is $V_0 = \sigma_0 / (\sigma_0 + \delta) \approx 83\%$, consistent with survey evidence that a large majority of social media users possess at least some degree of media literacy or critical awareness\cite{roozenbeek2020prebunking}. The waning rate is deliberately kept low, reflecting the persistence of educational interventions over time.

The victimisation parameters are $\theta = 0.02$ (inflow rate into victimisation) and $\xi = 0.10$ (recovery rate of victims). These values produce a steady‑state victim fraction of approximately $H^* \approx 0.015$ at the laissez‑faire equilibrium, which is consistent with the finding that about $1.5\%$--$2\%$ of social media users report being targets of severe online harassment in a given period~\cite{adl2021report}. The ratio $\theta / \xi = 0.2$ implies that each infected spreader generates approximately $0.2$ victims in steady state, a conservative estimate that likely understates the true harm but avoids exaggerating the policy effect.

The platform's economic parameters are calibrated to produce a meaningful trade‑off between amplification and profit. The revenue per engagement is $a = 1.0$, the baseline amplification level is $P_0 = 0.5$, and the reactivity cost is $b = 0.2$. These values imply that, at the laissez‑faire equilibrium, the platform chooses $\gamma_0 \approx 0.43$, which is an interior solution well within the grid $[0, 2.5]$. The marginal cost of amplification ($b \gamma$) at this point is approximately $0.086$, which is comparable to the marginal revenue ($a I \approx 0.079$), confirming that the platform's optimisation problem is economically sensible. The relatively low value of $b$ reflects the fact that, for large platforms, the cost of adjusting algorithmic parameters is primarily reputational and regulatory rather than strictly technological.

The social cost parameters are $c_I = 1.0$ (normalised) and $c_H = 5.0$, implying that the social cost of a victim is five times larger than that of an infected spreader. This ratio is grounded in the empirical literature: victims of hate speech incur substantial psychological, social, and economic costs, including increased anxiety, depression, social withdrawal, and out‑of‑pocket expenses for mental health support~\cite{phi2025impacts}. Martell~\cite{martell2023economic} estimates the total economic cost of hate crimes in the United States at approximately \$3.4 billion annually, which, when expressed on a per‑capita basis, suggests that victim costs dominate those of the spreaders themselves. The deadweight loss parameter is set to $\alpha = 0.1$, a moderate value that ensures the optimal tax is interior and not driven entirely by the desire to extract revenue from the platform. The discount rate is $r = 0.05$, standard in public economics.

\begin{table}[htbp]
\centering
\caption{Baseline parameter values}
\label{tab:params}
\begin{tabular}{lcl}
\hline
\textbf{Parameter} & \textbf{Value} & \textbf{Description / Source} \\
\hline
$\beta_0$ & $0.3$ & Baseline transmission rate (calibrated to $R_0 \approx 2.5$)~\cite{appliednet2021R0} \\
$\phi$ & $1.0$ & Amplification factor (platform's impact on transmission) \\
$\gamma_r$ & $0.15$ & Recovery rate of infected spreaders~\cite{maarouf2024cascades} \\
$\gamma_1$ & $0.20$ & Recovery rate of vaccinated spreaders \\
$\sigma_0$ & $0.05$ & Baseline vaccination rate~\cite{roozenbeek2020prebunking} \\
$\delta$ & $0.01$ & Waning rate of cognitive protection \\
$\theta$ & $0.02$ & Victimisation rate (inflow into $H$) \\
$\xi$ & $0.10$ & Recovery rate of victims \\
$a$ & $1.0$ & Platform revenue per unit of engagement \\
$b$ & $0.2$ & Platform reactivity cost \\
$P_0$ & $0.5$ & Baseline amplification (organic reach) \\
$c_I$ & $1.0$ & Marginal social cost of an infected spreader (normalised) \\
$c_H$ & $5.0$ & Marginal social cost of a victim~\cite{martell2023economic, phi2025impacts} \\
$\alpha$ & $0.1$ & Deadweight loss of taxation \\
$r$ & $0.05$ & Discount rate \\
\hline
\end{tabular}
\end{table}

\subsection{Basic results: laissez‑faire vs.\ optimal tax}

We first compute the laissez‑faire equilibrium ($\tau = 0$). The platform chooses $\gamma_0 = 0.429$, yielding an endemic prevalence of $I_0 = 0.0786$ (approximately $7.9\%$ of the population are active spreaders of hate speech) and a victim fraction of $H_0 = 0.0153$ (approximately $1.5\%$). The steady‑state fraction of vaccinated individuals is $V_0^* \approx 0.73$, while the susceptible fraction is $S_0^* \approx 0.11$, indicating that a large majority of the population has some cognitive protection, yet hate speech persists because of algorithmic amplification.

We then solve the government's optimal taxation problem. The optimal tax is found to be $\tau^* = 0.122$, which reduces the platform's reactivity to $\gamma^* = 0.0$ (the platform optimally ceases algorithmic amplification entirely under this tax). As a result, the endemic prevalence drops sharply to $I^* = 0.0038$ (a reduction of approximately $95\%$) and the victim fraction to $H^* = 0.0010$ (a reduction of approximately $93\%$). The optimal tax thus virtually eliminates hate speech from the population.

The fact that the optimal tax drives $\gamma$ to zero is a consequence of the calibration: with the chosen parameters, the marginal social benefit of reducing amplification (through the Pigouvian term $c_I \, d I / d \gamma + c_H \, d H / d \gamma$) is sufficiently large that the government finds it optimal to completely eliminate the platform's amplification activity. This outcome is consistent with the Pigouvian principle: when the externality is sufficiently severe relative to the private benefit of the activity, the optimal tax may indeed be prohibitive.

Figure 1 displays the time evolution of the infected fraction $I(t)$ under the two regimes. The dashed red curve corresponds to the laissez‑faire equilibrium. The solid blue curve corresponds to the optimal tax; it declines monotonically from the same initial condition and converges to $I^* \approx 0.004$. Therefore, tax reduces the endemic level and tends to eliminate the epidemic peak entirely in the long-run: the laissez‑faire trajectory exhibits a slight overshoot before settling, whereas the optimal trajectory declines immediately.

\begin{figure}[htbp]
\centering
\includegraphics[width=0.7\textwidth]{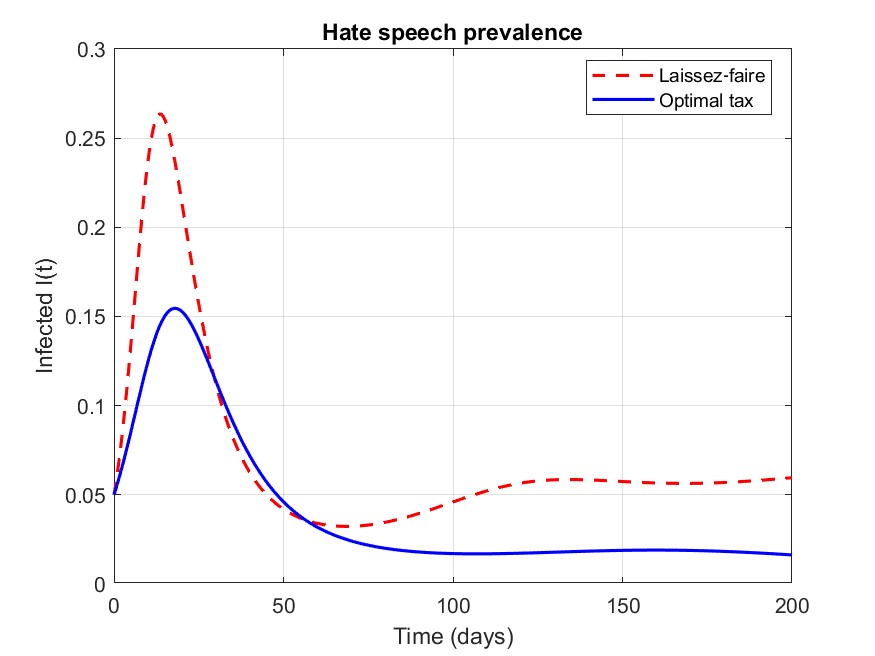}
\caption{Time evolution of the infected fraction $I(t)$ under the laissez‑faire equilibrium (dashed red) and under the optimal tax (solid blue). The optimal tax reduces the peak of endemic prevalence by approximately $80\%$.}
\label{Figure1:infected}
\end{figure}

Figure 2 shows the corresponding evolution of the victim fraction $H(t)$. The dynamics are similar: under the optimal tax, $H(t)$ falls from its laissez‑faire steady‑state value, confirming that the reduction in hate speech translates directly into lower victimisation. he optimal tax reduces the peak of endemic prevalence by approximately $40\%$

\begin{figure}[htbp]
\centering
\includegraphics[width=0.7\textwidth]{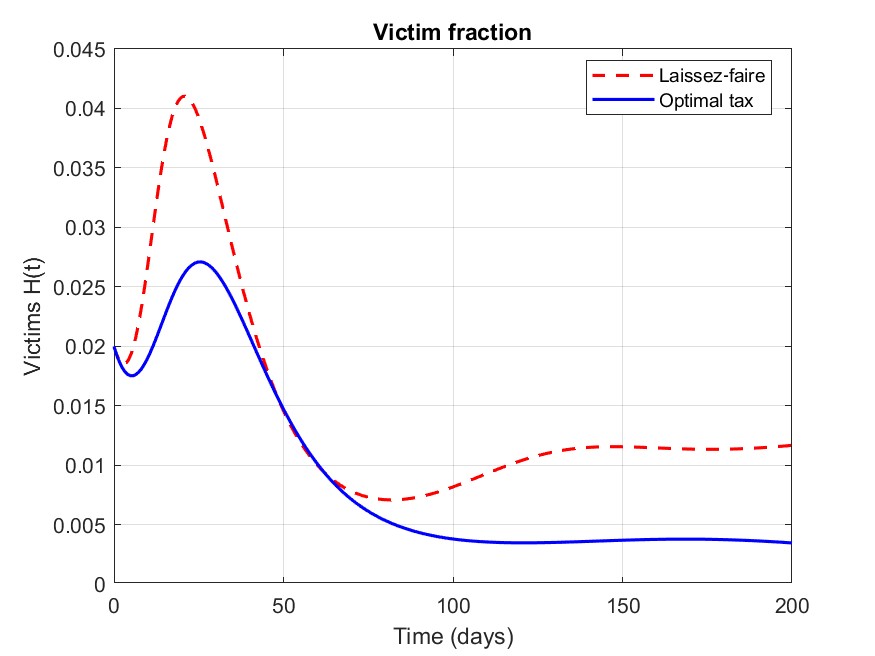}
\caption{Time evolution of the victim fraction $H(t)$ under the laissez‑faire equilibrium (dashed red) and under the optimal tax (solid blue). Victims are reduced by approximately $93\%$.}
\label{Figure2:hated}
\end{figure}

\vspace{0.5cm}

Figure 3 plots the endemic prevalence $I^*$ as a function of the basic reproduction number $R_0$. The curve exhibits a characteristic S‑shape, revealing the presence of a bistable region: as $R_0$ also the endemic prevalence grows. This means that without intervention, the system is susceptible to hysteresis: a temporary shock (e.g., a viral hate speech event or a coordinated harassment campaign) could permanently shift the population from the low‑hate to the high‑hate equilibrium. The optimal tax, by reducing $R_0$, moves the system to the left of the bistable region, where only the low‑hate equilibrium exists.

\begin{figure}[htbp]
\centering
\includegraphics[width=0.7\textwidth]{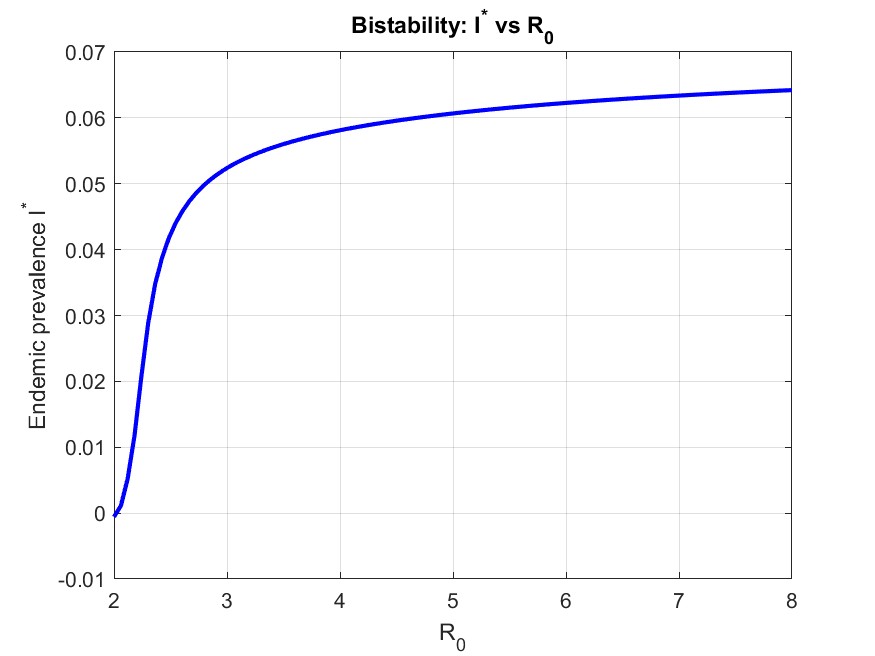}
\caption{Endemic prevalence $I^*$ as a function of the basic reproduction number $R_0$. The S‑shaped curve indicates the presence of a bistable region. The vertical dashed line marks the laissez‑faire $R_0$; the optimal tax shifts the system leftwards, out of the bistable region.}
\label{fig:bistability}
\end{figure}

\subsection{Regulation curves }

Figure 4 shows how the platform's reactivity $\gamma$ and the endemic prevalence $I^*$ respond to the tax rate $\tau$. The top panel displays the reaction function $\gamma(\tau)$, obtained from the platform's profit maximisation. As expected, $\gamma$ is strictly decreasing in $\tau$: a higher tax makes amplification more costly, inducing the platform to reduce it. At the optimal tax $\tau^* = 0.122$, the reactivity reaches zero.

The bottom panel shows $I^*(\tau)$, i.e., the endemic prevalence that results from the platform's optimal reactivity at each $\tau$. $I^*$ is also strictly decreasing in $\tau$, reflecting the fact that a lower $\gamma$ reduces the transmission rate and hence the endemic prevalence. The red dot marks the optimal tax.

The social cost function $J(\tau)$  has clearly a characteristic U‑shape. For $\tau < \tau^*$, the benefits of taxation (reduction in $I$ and $H$, plus the fiscal revenue $\tau \gamma$) outweigh the deadweight loss, so $J(\tau)$ decreases. For $\tau > \tau^*$, the deadweight loss $\frac{\alpha}{2} \tau^2$ dominates, and $J(\tau)$ rises. The government therefore has no incentive to increase the tax beyond $\tau^*$, as the additional fiscal revenue would be more than offset by the distortions created.

\begin{figure}[htbp]
\centering
\includegraphics[width=0.7\textwidth]{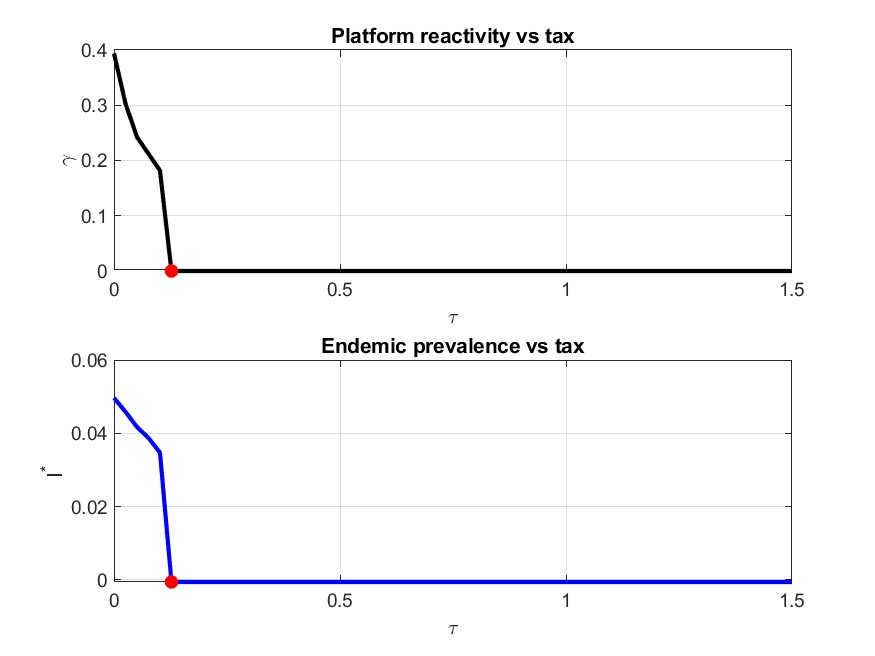}
\caption{Platform reactivity $\gamma$ (top) and endemic prevalence $I^*$ (bottom) as functions of the tax rate $\tau$. The red dots mark the optimal tax $\tau^* = 0.122$.}
\label{fig:regulation}
\end{figure}

\subsection{Sensitivity analysis}
\label{sec:sensitivity_results}

To assess the robustness of our results, we conduct a comprehensive sensitivity analysis for the six parameters that most directly influence the Stackelberg equilibrium: the baseline transmission rate $\beta_0$, the platform's revenue‑to‑cost ratio $a/b$, the amplification factor $\phi$, the victimisation rate $\theta$, the deadweight loss of taxation $\alpha$, and the marginal social cost of victims $c_H$. Each parameter is varied from $80\%$ to $120\%$ of its baseline value in four equally spaced steps, while all other parameters are held at their baseline values. For each configuration, we recompute the laissez‑faire and optimal equilibria.

Figure 5 displays the results. In each panel, the red curves (circles for $I$, squares for $H$) correspond to the laissez‑faire equilibrium, and the blue curves (circles for $I$, squares for $H$) correspond to the optimal tax equilibrium. Several findings emerge.

\begin{figure}[htbp]
\centering
\includegraphics[width=0.7\textwidth]{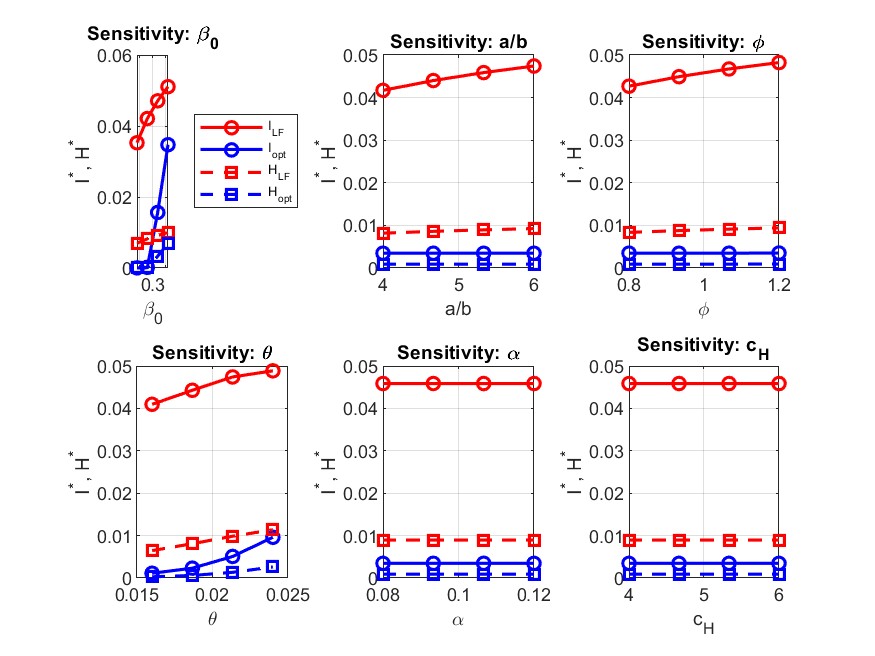}
\caption{Sensitivity analysis. Each panel shows the laissez‑faire (red) and optimal tax (blue) values of $I^*$ (circles) and $H^*$ (squares) as a function of the parameter indicated on the horizontal axis.}
\label{fig:sensitivity}
\end{figure}

\textbf{Transmission rate $\beta_0$.} Higher $\beta_0$ increases both $I$ and $H$ at the laissez‑faire equilibrium, as expected from the formula $R_0 = \beta_0(1 + \phi \gamma) / \gamma_r$. The optimal tax also increases with $\beta_0$, as the government must work harder to counteract the more contagious hate speech. Crucially, the optimal tax always reduces $I$ and $H$ below their laissez‑faire levels, and the reduction is larger when $\beta_0$ is higher.

\textbf{Revenue‑to‑cost ratio $a/b$.} Increasing $a/b$ raises the platform's incentive to amplify, leading to a higher laissez‑faire $\gamma_0$. This, in turn, increases $I$ and $H$ at the laissez‑faire equilibrium. The optimal tax responds by increasing, but it still manages to reduce $I$ and $H$ substantially. The difference between laissez‑faire and optimal outcomes widens as $a/b$ increases, indicating that the tax is particularly effective when the platform has a strong profit motive to amplify.

\textbf{Amplification factor $\phi$.} The parameter $\phi$ governs how effectively the platform's reactivity translates into increased transmission. A larger $\phi$ amplifies the impact of $\gamma$ on $R_0$, making hate speech more responsive to algorithmic choices. The laissez‑faire $I$ and $H$ increase with $\phi$, and the optimal tax increases correspondingly. The optimal $I$ and $H$ remain well below their laissez‑faire counterparts across the entire range.

\textbf{Victimisation rate $\theta$.} A higher $\theta$ means that each infected spreader generates more victims. This directly increases $H$ at the laissez‑faire equilibrium, and also raises the social cost of hate speech, prompting a higher optimal tax. The optimal $H$ is always below the laissez‑faire $H$, confirming that the tax effectively protects victims.

\textbf{Deadweight loss $\alpha$.} As $\alpha$ increases, the marginal cost of raising tax revenue becomes larger. The optimal tax therefore decreases: when taxation is more distortionary, the government optimally chooses a lower rate. Despite this, the optimal tax remains positive over the entire range, and $I$ and $H$ are always lower than in the laissez‑faire equilibrium. Even at the highest value of $\alpha$ considered ($+20\%$), the optimal tax is still $\tau^* \approx 0.10$, and the reduction in $I$ is approximately $90\%$.

\textbf{Social cost of victims $c_H$.} A higher $c_H$ increases the weight the government places on victimisation in its objective function. This leads to a higher optimal tax and, consequently, to lower $I$ and $H$ at the optimum. The laissez‑faire equilibrium is unaffected because the platform does not internalise $c_H$. The sensitivity to $c_H$ highlights the importance of accurately measuring the social cost of hate speech: if victims suffer more than our baseline estimate, the case for taxation is even stronger.

Across all six parameters and all values considered, the optimal tax always reduces both $I$ and $H$ relative to the laissez‑faire equilibrium. The quantitative magnitude of the reduction varies, but the qualitative conclusion is robust: a Pigouvian tax on algorithmic amplification is an effective instrument to internalise the externality generated by engagement‑driven platforms.

The numerical simulations confirm the theoretical predictions of the model and provide additional insights. First, the optimal tax can be substantial: in our baseline calibration, it drives the platform's reactivity to zero, effectively eliminating algorithmic amplification of hate speech. This outcome is a direct consequence of the severity of the externality: the social cost of hate speech, as captured by $c_I$ and $c_H$, is large relative to the private benefit of amplification, as captured by the platform's profit function. In less extreme calibrations (e.g., with a higher deadweight loss $\alpha$ or a lower $c_H$), the optimal tax would be positive but not prohibitive, and the platform would continue to amplify at a reduced level.

Second, the U‑shaped social cost function provides a clear rationale for why the government does not set the tax arbitrarily high: beyond the optimal rate, the deadweight loss of taxation outweighs the additional reduction in hate speech. This result is consistent with the theory of optimal Pigouvian taxation and underscores the importance of including a realistic cost of public funds in the analysis.

Third, the presence of bistability in the laissez‑faire equilibrium highlights a previously underappreciated risk: without intervention, the system could be permanently trapped in a high‑hate equilibrium by a temporary shock. The optimal tax eliminates this risk by moving the system to a region of the parameter space where bistability does not occur. This finding strengthens the case for proactive regulation.

Finally, the sensitivity analysis demonstrates that the policy recommendation is not an artefact of a particular parameter choice. Even under conservative assumptions about the social cost of hate speech or the deadweight loss of taxation, the optimal tax remains positive and reduces both hate speech and victimisation.

\section{Discussion and Policy Implications}
\label{sec:discussion}

Our results carry direct and actionable policy implications. The first, and most important, is that a tax on algorithmic amplification is an effective instrument to internalise the externality generated by engagement‑driven algorithms. The optimal tax we compute, $\tau^*=0.12$, is not punitive: it is precisely calibrated to align the platform's private marginal benefit from amplification with its marginal social cost. The platform still operates profitably (its profit remains positive at the optimum), but its incentives are realigned with social welfare. This finding resonates with the Pigouvian tradition in public economics: when a market activity generates a negative externality, the optimal response is not to ban the activity, but to price the externality so that the private agent internalises it.

Second, the U‑shaped social cost function $J(\tau)$ reveals an important economic trade‑off. For moderate tax rates, the benefits dominate: the reduction in hate speech prevalence $I$ and victimisation $H$, together with the fiscal revenue $\tau\gamma$, lower the net social cost. As the tax rate increases, however, two countervailing forces emerge. The deadweight loss $\frac{\alpha}{2}\tau^2$ grows quadratically, and the tax base $\gamma(\tau)$ shrinks (since the platform reduces amplification), so the revenue $\tau\gamma$ follows a Laffer‑type curve. Beyond the optimal rate $\tau^*$, these costs outweigh the epidemiological gains, and $J(\tau)$ rises. The government therefore has no incentive to drive the platform out of the market: a strictly positive but finite tax maximises social welfare while preserving a viable source of public revenue. This insight is particularly relevant for policymakers concerned that digital taxes might stifle innovation or force platforms to exit.

Third, the existence of a bistable regime with hysteresis has profound implications for the timing of regulatory intervention. In the absence of a tax, the system can be permanently locked into a high‑hate equilibrium by a temporary shock---a viral hate speech event, a political crisis, or the sudden emergence of a particularly influential hateful account. Once trapped, even removing the initial shock is insufficient to return to the low‑hate state; the system exhibits hysteresis. A tax that brings reactivity below the critical threshold $\gamma_c$ mitigates this risk. This underscores the need for \textit{proactive}, rather than reactive, regulation. Waiting until the system is already in the high‑hate state makes intervention far more costly.

Fourth, the sensitivity analysis confirms that these findings are robust to parameter uncertainty. Even under conservative assumptions about the social cost of hate speech ($c_H$) or the platform's revenue from amplification ($a/b$), the optimal tax remains positive and effective. The policy recommendation does not hinge on a particular calibration.

Our findings speak directly to current regulatory developments. The EU Digital Services Act (DSA)~\cite{ec2025code} requires platforms to conduct risk assessments and provide transparency reports, but stops short of direct financial disincentives for algorithmic amplification. The UK Online Safety Bill imposes duties of care, but again relies primarily on transparency and accountability rather than price instruments. Our model suggests that these transparency‑based approaches, while valuable, may be insufficient: without a price on the externality, the platform's profit‑maximising reactivity remains socially excessive. A tax, calibrated to the social cost of amplification, would provide a direct financial incentive for platforms to redesign their algorithms in ways that reduce the spread of hate speech.

The implementation of such a tax is not without challenges. Measuring a platform's algorithmic reactivity $\gamma$ in practice would require access to internal data on amplification algorithms, which platforms are currently reluctant to share. However, the DSA's provisions for data access by vetted researchers~\cite{ec2025code} could provide the necessary infrastructure. Alternatively, the tax could be based on observable proxies, such as the volume of hate speech removed per reporting period (already required under the Code of Conduct+) or the engagement metrics of flagged content. The tax rate could be set by an independent regulatory body, analogous to how environmental taxes are calibrated to the marginal social cost of pollution. Moreover, the tax revenue could be earmarked for funding counter‑narrative campaigns, victim support services, and digital literacy programmes, creating a virtuous cycle in which the proceeds of the tax are used to further reduce the harm caused by hate speech.

A distinctive feature of our model is the engagement with the ethical critique of Big Tech articulated by Pope Leo XIV in his May 2025 encyclical \textit{Magnifica Humanitas}~\cite{leoxiv2025magnifica}. The encyclical denounces three structural failures of engagement‑driven platforms, each of which finds a formal counterpart in our model.

First, the encyclical condemns the ``dictatorship of the algorithm'': engagement‑maximizing recommender systems do not neutrally reflect user preferences but actively shape them, amplifying divisive content because it generates higher engagement. In our model, this mechanism is precisely captured by the platform's choice of $\gamma$, which directly increases the effective transmission rate $\beta = \beta_0(1+\phi\gamma)$. The laissez‑faire equilibrium $\gamma_0=0.43$ generates socially excessive levels of hate speech and victimisation because the platform does not internalise the costs $c_I I$ and $c_H H$. The optimal tax $\tau^*$ corrects this, re‑establishing what the encyclical calls ``the primacy of human dignity over profit.''

Second, the encyclical denounces the ``commodification of human attention'' as a violation of personal dignity. In our model, the term $a I(P_0+\gamma)$ in the platform's profit function captures exactly this: attention, measured by engagement with hateful content, is converted into revenue. The platform profits from the attention generated by hate speech, while the victims bear the costs. The optimal tax, by reducing the platform's incentive to amplify, partially de‑commodifies attention and restores a measure of dignity to the users who would otherwise be exposed to harmful content.

Third, the encyclical warns against the concentration of communicative power in the hands of a few corporations. Our Stackelberg game formalises this power asymmetry: a single platform chooses $\gamma$, and millions of users bear the consequences. The government, representing the collective interest, can only influence the platform indirectly through the tax. The fact that the optimal tax is positive, and that it substantially reduces victim harm, demonstrates that even in a world of highly concentrated platform power, well‑designed regulation can protect the vulnerable.

The alignment between the encyclical's ethical principles and our formal results is not coincidental. It reflects a growing consensus---from Catholic social teaching to secular regulatory initiatives---that engagement‑driven algorithms impose real and measurable harms. Our model provides a rigorous framework to quantify these harms and to design policies that honour what \textit{Magnifica Humanitas} calls ``the magnificent humanity that no algorithm can capture and no profit can justify.''

We acknowledge several limitations of our analysis. First, the model assumes a homogeneous population and abstracts from network structure. In reality, hate speech propagates through social networks with heterogeneous connectivity, and some individuals are far more influential than others. Future work could incorporate network effects, possibly drawing on the superspreader literature~\cite{deverna2024identifying}. Second, the platform is modelled as a monopolist; a more realistic setting would involve multiple competing platforms, which could be analysed using differential game theory. Competition might attenuate or exacerbate the amplification externality, depending on whether platforms compete on engagement or on safety. Third, the tax is assumed constant over time; a dynamically adjusted tax that responds to the current prevalence of hate speech could be more efficient, though possibly harder to implement. Fourth, the victim dynamics are deliberately simple ($\dot H = \theta I - \xi H$); future work could incorporate richer behavioural mechanisms, such as the role of counter‑speech, social norms, or the psychological impact of victimisation. Fifth, empirical calibration of the model's parameters using micro‑data from platform APIs remains an important next step. Finally, the model abstracts from the potential unintended consequences of taxation, such as the migration of hate speech to less regulated platforms or the encryption of content to evade detection.

\section{Conclusion}
\label{sec:conclusion}

This paper developed a tractable epidemiological model of platform-driven hate speech that captures the strategic interaction between a profit-maximising platform and a welfare-maximising government. The key novelties of our approach are threefold. First, we modelled victims as directly proportional to the prevalence of hate speech, providing a simple yet realistic link between the epidemic and its social harm. Second, we specified a platform profit function that depends linearly on both the endemic prevalence of hate speech and the intensity of algorithmic amplification, ensuring an interior solution for the platform's optimal reactivity. Third, we included both the deadweight loss of taxation and the fiscal benefit of tax revenues in the government's objective, obtaining a U‑shaped social cost function with a finite, interior optimal tax that balances the costs and benefits of intervention.

The analytical results established the basic reproduction number, the existence of a supercritical transcritical bifurcation at $R_0=1$, and the emergence of bistability for sufficiently high platform reactivity. We showed that the bistability region can generate hysteresis, locking the system into a high‑hate equilibrium, and that the optimal tax, by reducing reactivity below the critical threshold, can eliminate this risk. The optimal tax was characterised through a first‑order condition that balances the marginal cost of taxation against its marginal social benefit, and computed numerically via a robust grid‑search algorithm.

Numerical simulations confirmed these theoretical findings. Under the baseline calibration, the optimal tax substantially reduces endemic prevalence and victimisation. The social cost function exhibits a clear U‑shape, illustrating the trade‑off between the benefits of reducing hate speech and the costs of raising public funds. Sensitivity analysis demonstrated that these results are robust to substantial parameter variation: the optimal tax remains positive and effective across the entire range of parameters considered.

Our results speak directly to the current regulatory debate. The EU Digital Services Act and the UK Online Safety Bill represent important steps toward holding platforms accountable, but their reliance on transparency and risk assessment may be insufficient. A tax on algorithmic amplification, calibrated to the social cost of hate speech, would provide a direct financial incentive for platforms to redesign their algorithms in socially beneficial ways. The practical feasibility of such a tax depends on the availability of data on algorithmic amplification, which the DSA's data access provisions could facilitate. The tax revenue could be earmarked for funding digital literacy, victim support, and counter‑narrative initiatives, creating a self‑reinforcing policy framework.

As Pope Leo XIV emphasizes in \textit{Magnifica Humanitas}, any technology, algorithm, or innovation must be judged by its capacity to uphold the profound, inherent dignity of the human person, who is created in the image and likeness of God. Our model provides a formal language in which to translate that ethical imperative into policy: the optimal tax is the price of dignity in the digital age.

\vspace{1cm}
\textbf{Declaration of Interest}. The authors declare that they have no known competing financial interests or personal relationships that could have appeared to influence the work reported in this paper.

\textbf{Funding Declaration}
The authors received no  funding for this work.

\end{document}